\documentclass[12pt]{article}

\usepackage{fullpage}
\usepackage{bbm,amsfonts,amsmath,amssymb,amsthm}

\usepackage{tikz}

\newtheorem{theorem}{Theorem}[section]
\newtheorem{definition}[theorem]{Definition}
\newtheorem{lemma}[theorem]{Lemma}
\newtheorem{corollary}[theorem]{Corollary}
\newtheorem{remark}[theorem]{Remark}
\newtheorem{example}[theorem]{Example}

\newtheorem{examples}[theorem]{Examples}

\newcommand{\sm}{\setminus}
\newcommand{\eset}{\emptyset}
\newcommand{\cM}{{\mathcal M}}
\newcommand{\cX}{{\mathcal X}}
\newcommand{\cC}{{\mathcal C}}
\newcommand{\TD}{\mathrm{TD}}
\newcommand{\RTD}{\mathrm{RTD}}
\newcommand{\PBTD}{\mathrm{PBTD}}
\newcommand{\VCD}{\mathrm{VCD}}
\newcommand{\seq}{\subseteq}
\newcommand{\ra}{\rightarrow}
\newcommand{\Span}[1]{\langle #1 \rangle}
\newcommand{\impl}{\Rightarrow}
\newcommand{\dund}{\Leftrightarrow}

\DeclareMathOperator{\order}{order}

\begin{document}

\title{RTD-Conjecture and Concept Classes \\ Induced by Graphs} 

\author{Hans Ulrich Simon}



\maketitle

\begin{abstract}
It is conjectured in~\cite{SZ2015} that the recursive teaching dimension 
of any finite concept class is upper-bounded by the VC-dimension 
of this class times a universal constant. In this paper, we confirm this
conjecture for two rich families of concept classes where each class
is induced by some graph $G$. For each $G$, we consider the class whose 
concepts represent stars in $G$ as well as  the class whose concepts 
represent connected sets in $G$. We show that, for concept classes of 
this kind, the recursive teaching dimension either equals the VC-dimension 
or is less by $1$.
\end{abstract}



\section{Introduction} \label{sec:introduction}

Teaching models explore the opportunities offered by learning from
carefully chosen examples (as opposed to random examples). 
Let $\cM$ be a model of teaching and let $\cC$ be a class of concepts.
The \emph{$\cM$-teaching dimension of $\cC$} is defined as the smallest
number $d$ such that every concept in $\cC$ can be inferred from $d$
(appropriately chosen) examples in accordance with the rules of the 
model~$\cM$.

In the original teaching model, as introduced by Shinohara and Miyano~\cite{SM1991}
resp.~by Goldman and Kearns~\cite{GK1995}, the examples used for teaching 
a concept $C$ must distinguish $C$ from any other concept in the underlying
concept class. This strong requirement, however, leads to intolerably large 
teaching dimensions even for very simple concept classes. The model of recursive 
teaching\footnote{See Section~\ref{subsec:teaching-models} for a definition.} 
had been suggested in~\cite{ZLHZ2011} as an  alternative (among others) to the 
original model. The corresponding teaching dimension is called the RT-dimension
and denoted by~$\RTD$. 

As demonstrated in~\cite{DFSZ2014} and~\cite{SSYZ2014}, 
the RT-dimension of a concept class $\cC$ is often closely 
related to the VC-dimension\footnote{The VC-dimension, first introduced 
in the pioneering work of Vapnik and Chervonenkis~\cite{VC1971}, plays
a central role in the theory of learning from random examples.} of $\cC$. 
The latter will be denoted by $\VCD(\cC)$ in the sequel.
The following is among the main open problems in connection with teaching:
\begin{description}
\item[RTD-Conjecture:] Is there a universal constant $c>0$
such that, for each finite concept class $\cC$, we have
that $\RTD(\cC) \le c \cdot \VCD(\cC)$?
\end{description}
It is known~\cite{HWLW2017} that $\RTD(\cC) = O\left(\VCD(\cC)^2\right)$.
It is furthermore known~\cite{CCT2016} that there exists a concept class $\cC$ 
such that $\RTD(\cC) \ge \frac{5}{3}\cdot\VCD(\cC)$.

We also briefly would like to mention that the RTD-conjecture is related 
to the well-known Sample-Compression conjecture of Warmuth~\cite{W2003}: 
\begin{itemize}
\item 
There is a one-to-one correspondence between so-called repetition-free
teaching sets in the RT-model and unlabeled sample compression schemes
having the acyclic non-clashing property. See Lemma~27 in~\cite{DFSZ2014}
for technical details.
\item
See also~\cite{DDSZ2013} for a connection between recursive teaching
and so-called order-compression schemes. The latter are special 
sample-compression schemes.
\end{itemize}
We cannot go into detail at this point. But it seems to be fair to say 
that progress with the RTD-conjecture is likely to produce progress
with the Sample-Compression conjeture, and vice versa.

Up to date, it is not known whether the RTD-conjecture can be confirmed in general.
We will show in this paper that it can be confirmed in a very strong sense
for two rich families of concept classes where each class is induced by 
some graph $G$. We consider classes whose concepts represent stars in $G$ 
as well as classes whose concepts represent connected sets 
in~$G$.\footnote{The VC-dimension of these classes has been studied
before in~\cite{KKRUW1997}.} We show that, for concept classes $\cC$ 
of this kind, it holds that
\[ \VCD(\cC)-1 \le \RTD(\cC) \le \VCD(\cC) \enspace . \]

\medskip\noindent
This paper is structured as follows. In Section~\ref{sec:definitions},
we fix some notation and call to mind some basic definitions (including
the definition of various teaching models). The main results on
graph-induced concept classes are stated and proved 
in Section~\ref{sec:graph-induced-classes}. The paper is closed
by mentioning some directions for future research.

\section{Definitions, Notations and Facts} \label{sec:definitions}

The powerset of a set $M$ is denoted as $2^M$.
For integers $i \le j$, we define $[i:j] = \{i,\ldots,j\}$,
i.e., $[i:j]$ denotes the interval of integers that range
from $i$ to $j$. Moreover, we set $[n] = [1:n]$.

\subsection{Graphs and Subgraphs} \label{subsec:graphs}

We first fix some notation for special graphs.
The complete graph with $n$ vertices is denoted by $K_n$.
Let $P_n$ denote a path of length $n$. Similarly, let $C_n$ 
denote a cycle of length $n$.

\begin{definition}[Open and Closed Neighborhood] \label{def:neighborhood}
Let $G=(V,E)$ be a graph and let $x \in V$ be a vertex in $G$.
Then $N_G^o(x) = \{y|\{x,y\} \in E\}$ denotes the set of vertices
that are adjacent to $x$ in $G$. $N_G^o(x)$ is said to be the 
{\em open neighborhood} of $x$ in $G$ 
whereas $N_G(x) = N_G^o(x)\cup\{x\}$ is called the 
{\em closed neighborhood} of $x$ in $G$. More generally, 
for $X \seq V$, we define $N_G(X) = \cup_{x \in X}N_G(x)$
and $N_G^o(X) = N_G(X) \sm X$.
\end{definition}

\begin{remark}
For all (not necessarily distinct) vertices $x,y \in V$, 
we have that $N_G^o(x) \neq N_G(y)$.
\end{remark}

\begin{proof}
Assume for sake of contradiction that $N_G^o(x) = N_G(y)$.
Since $y \in N_G(y)$, we also have $y \in N_G^o(x)$.
Thus $y$ is adjacent to $x$. On the other hand, since $x \notin N_G^o(x)$,
we have that $x \notin N_G(y)$. Thus $x$ is not adjacent to $y$.
We arrived at a contradiction.
\end{proof}

\begin{definition}[Subgraph-Relation]
Let $G = (V,E)$ be a graph. A graph $G' = (V',E')$ is a \emph{subgraph 
of $G$}, which is denoted as $G' \le G$, if $V' \seq V$ and $E' \le E$.
By $\Span{V'}_G$ we denote the subgraph of $G = (V,E)$ that is
\emph{spanned by $V' \seq V$}, i.e., it contains \emph{all} edges of $E$
whose endpoints are in $V'$.
\end{definition}

\subsection{Concept Classes and VC-Dimension} \label{subsec:concept-classes}

As usual a \emph{concept over domain $\cX$} is 
a function of the form $C: \cX \ra \{+,-\}$ or, 
equivalently\footnote{$C$ can be identified with $\{x\in\cX: C(x)=+\}$},
a subset of~$\cX$. The elements of $\cX$ are called \emph{instances}. 
If $C(x) = b \in \{+,-\}$, then $C$ is said to \emph{assign label $b$ 
to instance $x$}. A set whose elements are concepts over 
domain $\cX$ is referred to as a \emph{concept class over $\cX$}. 
A pair $(x,b) \in \cX\times\{+,-\}$ is called a 
(positive resp.~negative) example. 
A collection $S \seq \cX \times \{+,-\}$ of examples is called 
a \emph{sample}. A concept is said to be \emph{consistent 
with $S$} if $C(x)=b$ for for all $(x,b) \in S$.
The \emph{version space induced by $S$} is defined as the set 
of all concepts in $\cC$ that are consistent \hbox{with $S$}. 

\begin{description}
\item[General Assumption:] We will assume throughout the paper
that $\cX$ is finite.
\end{description}

\noindent
For the reader's convenience, we now call to mind the definition
of the VC-dimension:

\begin{definition}[Shattered Sets and VC-dimension~\cite{VC1971}]
Suppose that $\cC$ is a concept class over domain $\cX$.
A set $S \seq \cX$ is said to be \emph{shattered} by $\cC$
if, for any $S' \seq S$, there exists a concept $C \in \cC$
such that $S \cap C = S'$. The size of the largest subset
of $\cX$ that is shattered by $\cC$ is called the
\emph{VC-dimension} of $\cC$ and denoted by $\VCD(\cC)$.
\end{definition}

\noindent
The following holds for trivial reasons:
\begin{example}[VC-Dimension of the Powerset] \label{ex:powerset-vcd}
$\VCD(2^\cX) = |\cX|$.
\end{example}

Sometimes the VC-dimension of a concept class $\cC$ can be
easily expressed in terms of the VC-dimension of some subclasses:

\begin{remark}[Common Sense] \label{rem:max-vcd}
Let $\cC_1,\ldots,\cC_r$ be concept classes over pairwise disjoint
domains $\cX_1,\ldots,\cX_r$. Then, for $\cC = \cC_1 \cup\ldots\cup \cC_r$,
the following holds:
\begin{equation} \label{eq:max-vcd}
\VCD(\cC) = \max_{i=1,\ldots,r}\VCD(\cC_i) \enspace .
\end{equation}
\end{remark}

\begin{proof}
Equation~(\ref{eq:max-vcd}) holds simply because a shattered set 
cannot have a non-empty intersection with two or more of the 
domains $\cX_1,\ldots,\cX_r$. 
\end{proof}

\noindent
We close this section by noting that a concept class $\cC$ 
of VC-dimension $d$ satisfies
\begin{equation} \label{eq:sauer-vcd}
|\cC| \le \sum_{i=0}^{d}\binom{|\cX|}{i} \enspace ,
\end{equation}
where $\cX$ denotes the underlying domain.
This is a direct consequence of the well-known
Sauer-Shelah lemma~\cite{S1972,SH1972}.

\subsection{Teaching Models} \label{subsec:teaching-models}

Here is the formal definition of the original model of teaching dimension, 
as introduced in~\cite{SM1991,GK1995}: 

\begin{definition}[SM-GK-Model~\cite{SM1991,GK1995}] \label{def:gk-model}
Let $\cC$ be a concept class over $\cX$.  
A \emph{teaching set for $C \in \cC$} is a subset $D \seq \cX$
which distinguishes $C$ from any other concept in $\cC$, i.e.,
for each $C' \in \cC\sm\{C\}$, there exists an instance $x \in D$
such that $C(x) \neq C'(x)$. The size of the smallest teaching set
for $C \in \cC$ is denoted by $\TD(C,\cC)$. The \emph{teaching 
dimension of $\cC$ in the SM-GK-model of teaching} is then given by
\[ 
\TD(\cC) = \max_{C \in \cC} \TD(C,\cC) \enspace . 
\]
A related quantity is
\[ 
\TD_{min}(\cC) = \min_{C \in \cC} \TD(C,\cC) \enspace . 
\]
An \emph{SM-GK-teacher for $\cC$} is a mapping $T: \cC \ra 2^\cX$ 
such that, for each $C \in \cC$,
the set $T(C) \seq \cX$ is a teaching set for $C \in \cC$.
The \emph{order of $T$} is defined 
as $\order(T) = \max_{C \in \cC}|T(C)|$. 
\end{definition}

\noindent
With a mapping $T: \cC \ra 2^\cX$ and a concept $C \in \cC$, 
we associate the sample 
\[ S_{T,C} := \{(x,C(x)): x \in T(C)\} \enspace . \]
Obviously the following assertions are equivalent:
\begin{enumerate}
\item
$T$ is an SM-GK-teacher for $\cC$.
\item
For each concept $C \in \cC$: $C$ is the only concept that is 
consistent with $S_{T,C}$.
\item
For each $C \in \cC$: the version space induced by the sample $S_{T,C}$
consists of $C$ only.
\end{enumerate}
Intuitively, $S_{T,C}$ is what the teacher would present to
a learner. The learner can then infer the corresponding version 
space.

\begin{description}
\item[Warning:] 
We will occasionally be sloppy and blur the distinction between $T(C)$
and $S_{T,C}$. For instance, we may say that we add a positive 
(resp.~negative) example to the teaching set of $C$. 
The correct formulation would be: we add an instance $x$ 
such that $C(x)=1$ (resp.~$C(x)=0$) to $T(C)$ so that $(x,(C(x))$ 
is added to $S_{T,C}$.
\end{description}

\smallskip
The recursive teaching model (=RT-model) is among the models
that were suggested in~\cite{ZLHZ2011} as an alternative
to the classical SM-GK-model. Here is the definition:

\begin{definition}[RT-Model~\cite{ZLHZ2011}] \label{def:rt-model}
Let $\cC_{min} \seq \cC$ be the easiest-to-teach concepts in $\cC$, i.e., 
\[
\cC_{min} = \{C \in \cC: \TD(C,\cC) = \TD_{min}(\cC)\} \enspace .
\]
The \emph{RT-dimension of $\cC$} is given by
\[ 
\RTD(\cC) = \left\{ \begin{array}{ll}
              \TD_{min}(\cC) & \mbox{if $\cC = \cC_{min}$} \\
              \max\{\TD_{min}(\cC) , \RTD(\cC\sm\cC_{min})\} &
              \mbox{otherwise}
            \end{array} \right.
\enspace . 
\]  
\end{definition}

\smallskip
There is still one more teaching model, introduced in~\cite{GRSZ2017},
that we will need for technical reasons. It  brings into play a so-called 
preference relation. The idea behind this is to release the teaching set
for a concept $C \in \cC$ from the obligation to distinguish $C$ from 
concepts having a lower preference than $C$ itself. Here is the formal 
definition:

\begin{definition}[PBT-Model~\cite{GRSZ2017}] \label{def:pbt-model}
A strict partial ordering $\prec$ on $\cC$ is called a
\emph{preference relation on $\cC$}. We say that 
\emph{concept $C \in \cC$ is preferred over concept $C' \in \cC$
with respect to $\prec$} if $C' \prec C$. A \emph{teaching set
for $C \in \cC$ with respect to $\prec$} is a subset $D \seq \cX$
which distinguishes $C$ from any any other concept $C' \in \cC$ 
except possibly the ones having a lower preference than $C$.
In other words, for each concept $C' \in \cC\sm\{C\}$ 
that does not satisfy $C' \prec C$ (so that either $C \prec C'$
or $C$ and $C'$ are incomparable), there exists an instance $x \in D$ 
such that $C(x) \neq C'(x)$. The size of a smallest teaching set 
for $C \in \cC$ with respect to $\prec$ is denoted by $\TD_\prec(C,\cC)$.
The \emph{teaching dimension of $\cC$ with respect to $\prec$}
is then given by 
\[ \TD_\prec(\cC) = \max_{C \in \cC}\TD_\prec(C,\cC) \enspace . \]
The \emph{PBT-dimension\footnote{PBT = Preference-Based Teaching} 
of $\cC$} is given by
\[ \PBTD(\cC) = \min_\prec \TD_\prec(\cC) \enspace , \]
where $\prec$ ranges over all preference relations on $\cC$.
A \emph{PB-teacher for $\cC$} is a pair $(T,\prec)$ where $\prec$
is a preference relation on $\cC$ and $T$ is a mapping $T: \cC \ra 2^\cX$ 
such that, for each $C \in \cC$, the set $T(C) \seq \cX$ is a teaching set 
for $C$ with respect to $\prec$. The \emph{order of $T$} is defined 
as $\order(T) = \max_{C \in \cC}|T(C)|$. 
\end{definition}

\noindent
We briefly mention two quite natural choices for $\prec$:
 
\begin{definition}[Subset- and Superset-Preferences]
Let $\cC_0$ be a subclass of $\cC$. We say that we have 
\emph{subset-preferences (resp.~superset-preferences) on $\cC_0$} if, 
for all $C,C' \in \cC_0$, the inclusion $C \subset C'$
implies that $C$ is preferred over $C'$ (resp.~that $C'$ is preferred
over $C$).
\end{definition}

\noindent
Obviously the following holds:

\begin{remark} \label{rem:order-teacher}
If $(T,\prec)$ is a PB-teacher for $\cC$, 
then $\PBTD(\cC) \le \TD_{\prec}(\cC) \le \order(T)$.
Moreover, for each $T: \cC \ra 2^\cX$ and each preference 
relation $\prec$ on $\cC$, the following assertions are equivalent:
\begin{enumerate}
\item
$(T,\prec)$ is a PB-teacher for $\cC$.
\item
For each $C \in \cC$: $C$ is the unique most preferred
(with respect to $\prec$) concept in the version space 
induced by $S_{T,C}$.
\end{enumerate}
\end{remark}

\medskip\noindent
The above models are related as follows:
\begin{enumerate}
\item
It was shown in~\cite{DFSZ2014} that 
\begin{equation} \label{eq:rtd-tdmin}
\RTD(\cC) = \max_{\cC_0 \seq \cC}\TD_{min}(\cC_0) \enspace , 
\end{equation}
holds for every concept class $\cC$. This implies that, for every concept class $\cC$, 
we have $\RTD(\cC) \ge \TD_{min}(\cC)$.
\item
It was shown in~\cite{GRSZ2017} that $\RTD$ and $\PBTD$
coincide on finite concept classes. 
\end{enumerate}

\begin{example} [RT-Dimension of the Powerset] \label{ex:powerset-rtd}
It is obvious that $\RTD(2^\cX) \le |\cX|$ and $\TD_{min}(2^\cX) = |\cX|$.
Since $\RTD$ is lower-bounded by $\TD_{min}$, 
we have $\RTD(2^\cX) \ge |\cX|$ and, therefore, $\RTD(2^\cX) = |\cX|$.
\end{example}

\noindent
Here is a weak pendant to Remark~\ref{rem:max-vcd}:

\begin{remark}[Common Sense] \label{rem:max-rtd}
Let $\cC_1,\ldots,\cC_r$ be concept classes over pairwise disjoint
domains $\cX_1,\ldots,\cX_r$. Then, for $\cC = \cC_1 \cup\ldots\cup \cC_r$,
the following holds:
\begin{equation} \label{eq:max-rtd} 
\max_{i=1,\ldots,r}\RTD(\cC_i) \le \RTD(\cC) \le 1+\max_{i=1,\ldots,r}\RTD(\cC_i)
\enspace . 
\end{equation}
\end{remark}

\begin{proof}
The first inequality holds for reasons of monotonicity.  The second-one 
holds because we can always add a positive example to the teaching set 
of a non-empty concept in $\cC$ (which is necessary only when the original 
teaching set consisted of negative examples only). Moreover, if $\eset \in \cC$,
we may consider $\eset$ as the most preferred concept in $\cC$.
We omit the (more or less) straightforward details.
\end{proof}

\noindent
Here comes a kind of RTD-version of the Sauer-Shelah Lemma:

\begin{remark}[\cite{SSYZ2014}] \label{rem:sauer-rtd}
The upper bound on $|\cC|$ that is shown in~(\ref{eq:sauer-vcd}) is valid
also for $d = \RTD(\cC)$. 
\end{remark}

\noindent
This gives a weapon for lower-bounding the $\RTD$:
\begin{equation} \label{eq:sauer-rtd}
|\cC| > \sum_{i=0}^{d}\binom{|\cX|}{i} \impl \RTD(\cC)>d \enspace .
\end{equation}

\noindent
In the course of the paper, we will proceed as follows:
\begin{itemize}
\item
We will often use equation~(\ref{eq:rtd-tdmin}) for deriving a lower bound 
on $\RTD(\cC)$ by choosing a hard-to-teach subclass $\cC_0 \seq \cC$.
Occasionally, a lower bound will be directly obtained 
from~(\ref{eq:sauer-rtd}).
\item
Because the RTD-conjecture is less known under the name PBTD-conjecture, 
we will state the results in terms of the $\RTD$. On the other hand,
we derive upper bounds on the $\RTD$ by switching to the 
equivalent\footnote{for finite concept classes} PBT-setting. 
More precisely, we obtain an upper bound on $\RTD(\cC)$ by designing 
a PB-teacher $(T,\prec)$ and by exploiting the facts 
that $\order(T)$ upper-bounds $\PBTD(\cC)$ and $\PBTD(\cC) = \RTD(\cC)$. 
\end{itemize}

\section{Teaching Complexity of Classes Defined by Graphs}
\label{sec:graph-induced-classes}

\noindent
As it was done in~\cite{KKRUW1997}, we associate 
the following concept classes with a given graph $G = (V,E)$: 
\begin{eqnarray*}
\cC_{star}(G) & = & \{X \cup \{x\}:\ x \in V , X \seq N_G^o(x)\} \\
\cC_{con}(G) & = & 
\{X \seq V:\mbox{$G$ has a subgraph that is a tree with vertex set $X$}\} \\
& = &\{X \seq V:\mbox{$X \neq \eset$ and $\Span{X}_G$ is connected}\}
\end{eqnarray*}
We will refer to sets in $\cC_{con}(G)$ as \emph{connected sets (in $G$)}.
We would like to stress that the empty set is considered as connected,
i.e., $\eset \in \cC_{con}(G)$ for every graph $G$.

The main goal in this section is to show that, for every graph $G$
and $\cC$ being either $\cC_{star}(G)$ or $\cC_{con}(G)$, the
RT-dimension and the VC-dimension of $\cC$ differ by at most $1$ and 
the RT-dimension of $\cC$ is upper-bounded by the VC-dimension of $\cC$.

Throughout this section, the symbols $G$, $V$ and $E$ 
will denote a graph $G$ with vertex set $V$ and edge set $E$.
Note that $V$ can also be viewed as the set of instances of 
either $\cC_{star}(G)$ or $\cC_{con}(G)$.

\subsection{Stars as Concepts} \label{subsec:stars-concepts}

As usual $\Delta(G) = \max_{x \in V}\deg_G(x)$ denotes the maximum
vertex-degree in $G$. Note that, for any $x \in V$,
the set $N_G^o(x)$ is shattered by $\cC_{star}(G)$. On the other hand,
if $S \seq V$ is shattered by $\cC_{star}(G)$, then $S \seq N_G(x)$
for some vertex $x \in X$ (because there must be a concept in $\cC_{star}(G)$
that assigns label $+$ to all vertices in $S$). Hence the
following holds:

\begin{theorem}[\cite{KKRUW1997}] \label{th:star-vcd-bounds}
$\Delta(G) \le \VCD(\cC_{star}(G)) \le \Delta(G)+1$.
\end{theorem}

\noindent
A similar result holds with RTD at the place of VCD:

\begin{theorem} \label{th:star-rtd}
$\Delta(G) \le \RTD(\cC_{star}(G)) \le \Delta(G)+1$.
\end{theorem}

\begin{proof}
Pick a vertex $x_0 \in V$ of degree $\Delta(G)$. 
Then $\cC_0 = \{ X \cup \{x_0\}:\ X \seq N_G^o(x_0) \}$ is a subclass
of $\cC_{star}(G)$ over domain $V$. The instances outside $N_G^o(x_0)$
do not distinguish between distinct concepts in $\cC_0$ and are therefore
useless for the purpose of teaching. But $\cC_0$ restricted to
subdomain $N_G^o(x_0)$ coincides with the powerset of $N_G^o(x_0)$.
It follows that 
\[
\RTD(\cC_{star}(G)) \ge \RTD(\cC_0) = |N_G^o(x_0)| = \deg_G(x_0) = 
\Delta(G) \enspace ,
\]
where the first inequality follows from a monotonicity argument.
As for the upper bound on $\RTD(\cC_{star}(G)) = \PBTD(\cC_{star}(G))$, 
we argue as follows. Assuming subset-preferences, each concept can be 
taught by presenting all its elements as positive examples. 
Hence $\RTD(\cC_{star}(G)) \le \max\{|C|: C \in \cC_{star}(G)\} = 1+\Delta(G)$, 
which concludes the proof.
\end{proof}

Let $V_{max} = \{x \in V: \deg_G(x) = \Delta(G)\}$ be the set of
vertices of maximum degree in $G$. We consider the following 
equivalence relation on $V_{max}$:
\[ x \equiv y \dund N_G(X) = N_G(y) \enspace . \]
In other words, two vertices of maximum degree are considered 
as equivalent iff they have the same closed neighborhood in $G$.
Note that the equivalence $x \equiv y$ of two distinct 
vertices $x \neq y \in V_{max}$ implies that $x$ and $y$ are adjacent. 
Let $V_{max} = V_1 \cup\ldots\cup V_m$ denote
the partition of $V_{max}$ into equivalence classes.
For $i=1,\ldots,m$, let $N_i \seq V$ be the common closed neighborhood
of the vertices in $V_i$. Set $V'_i = N_i \sm V_i$. The following 
observations are rather obvious: 
\begin{itemize}
\item $|N_i| = \Delta(G)+1$.
\item 
$V_i \seq N_i$ so that $N_i = V_i \cup V'_i$ is a partition of $N_i$.
\item
The vertices in $V_i$ form a clique in $G$. 
\item
The condition 
\begin{equation} \label{eq:big-vcd}
V'_i \seq N_G(v) 
\end{equation}
is satisfied by each $v \in V_i$ but violated 
by each $v \in V'_i$.\footnote{A vertex $v \in N_i$ with $V'_i$ in its
closed neighborhood is easily seen to have $N_i$ in its closed
neighborhood. Becsue $|N_i| = \Delta(G)+1$, we even have $N_G(v) = N_i$.
Thus $v$ must be a vertex in $V_i$ (and does therefore not belong to $V'_i$).}
\end{itemize}

\begin{theorem} \label{th:star-vcd-exact}
With the above notation, the following holds:
$\VCD(\cC_{star}(G)) = \Delta(G)+1$ iff there exists
an index $i \in [m]$ and a vertex $v \notin N_i$
that satisfies Condition~(\ref{eq:big-vcd}).
\end{theorem}

\begin{proof}
We start with the if-direction. Since $|N_i| = \Delta(G)+1$,
it suffices to show that $N_i$ is shattered by $\cC_{star}(G)$. 
Pick a subset $S \seq N_i$ of $N_i$. We proceed by case analysis:
\begin{description}
\item[Case 1:] $S$ contains an element $x$ from $V_i$. \\ 
Then, writing $S = \{x\} \cup (S \sm \{x\})$, we see that $S$ is a 
member of $\cC_{star}(G)$.  Moreover, since $S \seq N_i$, we 
have $S = S \cap N_i$.
\item[Case 2:] $S \cap V_i = \eset$. \\
Then $S \seq V'_i$ and the concept $\{v\} \cup S$ is a member
of $\cC_{star}(G)$. Moreover, since $v \notin N_i$, we 
have $(\{v\} \cup S) \cap N_i = S$.
\end{description}
This discussion shows that $N_i$ is shattered. \\
We continue with the only-if direction. Pick a set $S$ of 
size $\Delta(G)+1$ that is shattered. There must exist a concept
in $\cC_{star}(G)$ which contains $S$ resp.~even equals $S$
(since no concept in $\cC_{star}(G)$ has a size greater 
than $\Delta(G)+1$). This is possible only if $S$ is the closed
neighborhood of some vertex of maximum degree in $G$, i.e.,
if $S = N_i$ for some $i \in [m]$. There must exist a concept
in $\cC_{star}(G)$ whose intersection with $N_i$ equals $V'_i$.
Thus there must exist a vertex $v \notin V_i$ which satisfies 
the condition~(\ref{eq:big-vcd}). As already observed above,
a vertex satisfying this condition cannot belong to $V'_i$.
Hence $v \notin N_i$, which completes the proof.
\end{proof}

\begin{theorem} \label{th:star-rtd-vcd}
For any graph $G$: $\RTD(\cC_{star}(G)) \le \VCD(\cC_{star}(G))$.
\end{theorem}

\begin{proof}
In view of Theorems~\ref{th:star-vcd-bounds},~\ref{th:star-rtd}
and~\ref{th:star-vcd-exact} it suffices to show 
that $\RTD(\cC_{star}(G)) \le \Delta(G)$ under the following
assumption: for each $i \in [m]$, the condition~(\ref{eq:big-vcd})
is satisfied only by vertices $v \in N_i$ and therefore\footnote{since
this condition is violated by all $v \in V'_i$} only by 
vertices $v \in V_i$. Call a concept in $\cC_{star}(G)$ \emph{special}
if it contains $V'_i$ for some $i \in [m]$. Given the above assumption,
a concept of this kind is necessarily of the form $S \cup V'_i$
for some non-empty $S \seq V_i$. In order to show that $\Delta(G)$ 
upper-bounds $\RTD(\cC_{star}(G)) = \PBTD(\cC_{star}(G))$,
we proceed as follows:
\begin{itemize}
\item 
Non-special concepts are preferred over special ones.
\item 
Among the special concepts, say the ones containing $V'_i$,
we prefer the ones with the largest number of vertices in $V_i$.
\item
On the non-special concepts, we have subset-preferences.
\item
A non-special concept is taught by presenting all its vertices 
as positive examples.
\item
A special concept of the form $S \cup V'_i$ for some $i \in [m]$
and some non-empty $S \seq V_i$ is taught by presenting all vertices 
from $V'_i$ as positive examples and all vertices from $V_i \sm S$
as negative examples.
\end{itemize}
In any case, the concept $C$ that is to be taught is the unique most preferred 
concept in the respective version space. A teaching set for a special concept
containing $V'_i$ has size at most $\Delta(G)$ because $|N_i| = \Delta(G)+1$
and the teaching set is of the form $N_i \sm S$ for some non-empty $S \seq N_i$.
A teaching set for a non-special concept $C$  is also of size at most $\Delta(G)$:
either $C$ is a subset of the closed neighborhood of a vertex of non-maximal
degree or $C$ is of the form $S \cup T \seq N_i$ with $\eset \neq S \seq V_i$
and $T \subset V'_i$ (so that $|C| = |S|+|T| \le |V_i|+|V'_i|-1 \le \Delta(G)$).
\end{proof}

\noindent
The preceding results imply that
\begin{equation} \label{eq:star} 
\Delta(G) \le \RTD(\cC_{star}(G)) \le \VCD(\cC_{star}(G)) \le \Delta(G)+1
\enspace .
\end{equation}
Exactly one of these inequalities must be strict. The following examples
demonstrate that the strict inequality can occur in any of the three 
possible positions:

\begin{examples} \label{ex:star}
The \emph{$(\Delta,\RTD,\VCD)$-triple} of a graph $G$ is defined 
as the triple 
\[ 
(\Delta(G), \RTD(\cC_{star}(G)), \VCD(\cC_{star}(G)))
\enspace .
\]
Let us first consider the complete graph $K_n$, say with vertices $1,\ldots,n$. 
Clearly $\Delta(K_n) = n-1$ and $\cC_{star}(K_n) = 2^{[n]} \sm \{\eset\}$.
It is obvious that this class has VC-dimension $n-1$.
According to~(\ref{eq:star}), the RT-dimension also equals $n-1$.
Thus the $(\Delta,\RTD,\VCD)$-triple of $K_n$ equals $(n-1,n-1,n-1)$. \\
Consider now the path $P_n$ of length $n \ge 2$ with vertices $1,\ldots,n$
(in this order). Clearly $\Delta(P_n) = 2$ and $\cC_{star}(P_n)$
consists of all sets of $1$, $2$ or $3$ consecutive vertices on $P_n$.
It is obvious that this class has VC-dimension $2$ and,
again by virtue of~(\ref{eq:star}), the RT-dimension equals $2$ as well.
Thus the $(\Delta,\RTD,\VCD)$-triple of $P_n$ equals $(2,2,2)$. \\
Let $C_4$ be the cycle of length $4$ with vertices $a,b,c,d$ as
shown in the left part of Fig.~\ref{fig:star}. Clearly $\Delta(C_4) = 2$
and $\cC_{star}(C_4) = 2^{\{a,b,c,d\}} \sm \{ \eset,\{a,c\},\{b,d\},\{a,b,c,d\} \}$.
This class has a VC-dimension of at most $3$. Moreover it 
has $16-4 = 12 > 11 = 1 + 4 + \binom{4}{2}$ concepts. 
From~(\ref{eq:sauer-rtd}), we infer that $\RTD(\cC_{star}(C_4)) \ge 3$.
In combination with~(\ref{eq:star}), we may conclude that the
$(\Delta,\RTD,\VCD)$-triple of $C_4$ equals $(2,3,3)$. \\
The $(\Delta,\RTD,\VCD)$-triple of $C_n$ with $n \ge 5$ equals $(2,2,2)$.
The reason for that is that the VC-dimension of $\cC_{star}(C_n)$
equals $2$ for $n \ge 5$. Loosely speaking: from a local perspective 
(by considering only path segments of length at most $3$ on $C_n$),
there is not so much of a difference between $C_n$ and $P_n$. \\
Consider now the graph $G$ with $\Delta(G)=3$ on the right side 
of Fig.~\ref{fig:star}. The vertices $a$ and $b$ have the same closed 
neighborhood, namely $\{a,b,c,d\}$. In accordance with 
Theorem~\ref{th:star-vcd-exact}, the vertex $v$ demonstrates 
that the $\VCD(\cC_{star}(G)) = \Delta(G)+1=4$. An RTD-value 
(= $\PBTD$-value) of $3$ can be achieved as follows:
\begin{itemize}
\item 
Prefer concepts not containing $v$ over concepts containing $v$.
\item
On the concepts containing $v$, we have subset-preferences.
\item
On the concepts not containing $v$, use superset-preferences.
\item
Teach a concept containing $v$ by presenting all its vertices as
positive examples.
\item
Teach a concept $C$ not containing $v$ (and therefore containing at least
one of the vertices $a,b,c,d$) by presenting all vertices 
in $\{a,b,c,d\} \sm C$ as negative examples.
\end{itemize}
We conclude from this discussion that the $(\Delta,\RTD,\VCD)$-triple
of $G$ equals $(3,3,4)$.
\end{examples}

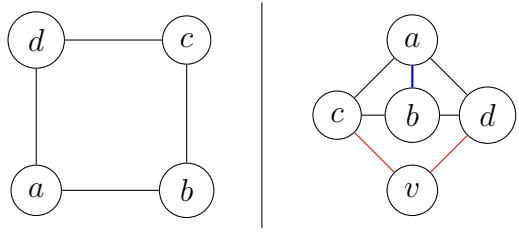
\begin{figure}[hbt]
	\begin{center}
\begin{tikzpicture}


	\tikzset{every node/.style={draw,circle}}

	\def\x{0} \def\y{0}

        \node (a) at (\x,\y-1) {$a$};
        \node (b) at (\x+2,\y-1) {$b$};
        \node (c) at (\x+2,\y+1) {$c$};
	\node (d) at (\x,\y+1) {$d$};

        \draw (a)--(b);
        \draw (b)--(c);
        \draw (c)--(d);
        \draw (d)--(a);

        \draw (\x+3,\y+1.5) -- (\x+3,\y-1.5);

        \def\x{4}
        \node (a) at (\x+1,\y+1) {$a$};
        \node (b) at (\x+1,\y) {$b$};
        \node (c) at (\x,\y) {$c$};
        \node (d) at (\x+2,\y) {$d$};
        \node (v) at (\x+1,\y-1) {$v$};

        \draw[blue,thick] (a)--(b);
        \draw (b)--(c);
        \draw (b)--(d);
        \draw (a)--(c);
        \draw (a)--(d);
        \draw[red] (c)--(v);
        \draw[red] (d)--(v);

\end{tikzpicture}
	\end{center}	
	\caption{A graph with $(\Delta,\RTD,\VCD)$-triple $(2,3,3)$
          and a second graph with $(\Delta,\RTD,\VCD)$-triple $(3,3,4)$. 
          \label{fig:star} }
\end{figure}

\subsection{Connected Sets as Concepts} \label{connected-sets-concepts}

For a connected graph $G$, let $\ell(G)$ denote the number of 
leaves of a maximum-leaf spanning tree of $G$. For a graph $G$
with connected components $G_1,\ldots,G_r$, we define
\begin{equation} \label{eq:max-ell}
\ell(G) = \max_{i=1,\ldots,r}\ell(G_i) \enspace .
\end{equation}

\noindent
The following observations are rather obvious.

\begin{remark} \label{rem:components}
Suppose that $G = (V,E)$ is a graph with components $G_i = (V_i,E_i)$ 
for $i = 1,\ldots,r$ where $r\ge2$. Then 
\begin{equation} \label{eq:components} 
\cC_{con}(G) = \bigcup_{i=1}^{r}\cC_{con}(G_i)
\end{equation}
and the domains $V_1,\ldots,V_r$ of $\cC_{con}(G_1),\ldots,\cC_{con}(G_r)$ 
are pairwise disjoint. Moreover 
\begin{eqnarray} 
\VCD(\cC_{con}(G)) & = & \max_{i=1,\ldots,r}\VCD(\cC_{con}(G_i)) 
\label{eq:components-vcd} \\
\max_{i=1,\ldots,r}\RTD(\cC_{con}(G_i)) & \le & \RTD(\cC_{con}(G)) \le 
\max_{i=1,\ldots,r}\RTD(\cC_{con}(G_i)) \label{eq:components-rtd} \enspace . 
\end{eqnarray}
\end{remark}

\begin{proof}
Equation~(\ref{eq:components}) simply holds because 
a set $X \in \cC_{con}(G)$ cannot have a nonempty intersection
with two or more components of $G$. Then~(\ref{eq:components-vcd})
and~(\ref{eq:components-rtd}) are a direct consequence of the corresponding 
equation resp.~inequalities in~(\ref{eq:max-rtd}).
\end{proof}

\noindent 
The following observations in Remark~\ref{rem:con} and in 
Lemma~\ref{lem:tree2spanning-tree} are quite obvious:

\begin{remark} \label{rem:con}
\begin{enumerate}
\item For any graph $G$ and any subgraph $G' \le G$:
      $\cC_{con}(G') \seq \cC_{con}(G)$.
\item For any tree $T$, $\ell(T)$ equals the number of leaves in $T$.
\item Let $T'$ be a connected subgraph of a tree $T$.
      Then $T'$ is a tree itself and $\ell(T') \le \ell(T)$.
      Moreover any connected subgraph of $T$ which contains all leaves 
      of $T'$ must necessarily contain all vertices of $T'$.
\end{enumerate}
\end{remark}

\begin{lemma} \label{lem:tree2spanning-tree}
For any connected graph $G$, the following hold:
\begin{enumerate}
\item
For any $X \in \cC_{con}(G)$, there exists a spanning tree 
for $\Span{N_G(X)}_G$ which has the vertices in $N_G^o(X)$
among its leaves.
\item
There exists a set $X \in \cC_{con}(G)$ such that $|N_G^o(X)| = \ell(G)$.
\item
Any tree $T \le G$ with $k$ leaves can be extended to a spanning tree $T^*$
of $G$ with at least $k$ leaves. Moreover, if $k = \ell(G)$, then the extension $T^*$
of $T$ can be chosen such that the following holds:
\begin{enumerate}
\item 
$T^*$ consists of $T$ and additional paths starting at the leaves of $T$.
\item
If there exists an interior vertex $u$ of $T$ such that each pair 
of distinct leaves in $T$ can be connected by a path in $G$ which does neither
contain $u$ nor any other leaf in $T$, then each pair of distinct leaves in $T^*$ 
can be connected by a path in $G$ which does neither contain $u$ nor any other
leaf in $T^*$.
\end{enumerate}
\end{enumerate}
\end{lemma}

\begin{proof}
\begin{enumerate}
\item
The desired spanning tree $T$ for $\Span{N_G(X)}_G$ can be built as follows.
Initialize $T$ as an (arbitrarily chosen) spanning tree for $\Span{X}_G$. Then,
for all $y \in N_G^o(X)$, do the following: pick a vertex $x \in X$
which is adjacent to $y$ and insert the edge $(x,y)$ into $T$ (so that $y$
becomes a new leaf).
\item
A suitable set $X$ is the set of interior vertices of a maximum-leaf 
spanning tree $T$ for $G$ (so that $N_G^o(X)$ is the set of $T$'s leaves).
\item
Suppose that $T \le G$ is a tree with $k$ leaves. If $T$ is not 
a spanning tree, then there exists an edge $e$ that connects a
vertex $v$ in $T$ with a vertex $v'$ not in $T$. Extend $T$ by adding
edge $e$ and let $T'$ denote the new tree. The new vertex $v'$ is
a leaf of $T'$. Moreover, each leaf in $T$, with the possible 
exception of $v$, is still a leaf in $T'$. Hence $T'$ has $k$ 
or $k+1$ leaves. We max proceed in this manner until the iteratively
extended tree is a spanning tree $T^*$ of $G$. Suppose now that $k = \ell(G)$.
Then $T'$, consisting of $T$ and the additional edge $\{v,v'\}$,
has exactly $k$ leaves. It necessarily follows that $v$ is a leaf 
in $T$ but not a leaf in $T'$. Using this argument iteratively, we
see that the final spanning tree $T^*$ consists of $T$ and additional paths
that start at the leaves of $T$. Finally suppose that there exists 
an interior vertex $u$ of $T$ such that each pair $v,w$ of distinct 
leaves in $T$ can be connected by a path $P_{v,w}$ in $G$ which does 
neither contain $u$ nor any other leaf in $T$. Consider a pair $v',w'$ 
of distinct leaves in $T^*$. Let $v,w$ be the corresponding pair of leaves 
in $T$ (such that $T^*$ contains a path $P_{v',v}$ from $v'$ to $v$ and 
another path $P_{w,w'}$ from $w$ to $w'$). Then the composition of the 
paths $P_{v',v} , P_{v,w} , P_{w,w'}$ is a path in $G$ that connects $v'$
and $w'$ and does neither contain $u$ nor any other leaf in $T^*$.
\end{enumerate}
\end{proof}

\noindent 
As a direct consequence of Lemma~\ref{lem:tree2spanning-tree},
we get the following result:

\begin{corollary} \label{cor:tree2spanning-tree}
For any connected graph $G$, the parameter $\ell(G)$ coincides 
with the largest number $k$ such that there exists a tree $T \le G$ 
with $k$ leaves. Moreover it also coincides with the size of the 
largest open neighborhood of some set in $\cC_{con}(G)$.
\end{corollary}

\begin{lemma}[\cite{KKRUW1997}] \label{lem:con-tree}
For any tree $T$: $\VCD(\cC_{con}(T)) = \ell(T)$.
\end{lemma}

\begin{theorem}[\cite{KKRUW1997}] \label{th:con-vcd}
For any connected graph $G$: $\ell(G) \le \VCD(\cC_{con}(G)) \le \ell(G)+1$.
\end{theorem}

\begin{corollary}
For any graph $G$: $\ell(G) \le \VCD(\cC_{con}(G)) \le \ell(G)+1$.
\end{corollary}

\begin{proof}
This is a direct consequence of Theorem~\ref{th:con-vcd}
in combination with the equations~(\ref{eq:max-ell})
and~(\ref{eq:max-vcd}). The application of~(\ref{eq:max-vcd})
is justified by the first part of Remark~\ref{rem:components}.
\end{proof}

\noindent
Similar results hold with RTD at the place of VCD:

\begin{theorem} \label{th:con-rtd}
\begin{enumerate}
\item
For any tree~$T$: $\RTD(\cC_{con}(T)) = \ell(T)$.
\item
For any connected graph $G$: $\RTD(\cC_{con}(G)) \ge \ell(G)$.
\item
For any graph $G$: $\ell(G) \le \RTD(\cC_{con}(G)) \le \ell(G)+1$.
\end{enumerate}
\end{theorem}

\begin{proof}
\begin{enumerate}
\item
We first show that $\ell(T)$ bounds $\RTD(\cC_{con}(T))$ from below.
Let $X_0$ be the set of interior vertices of $T$ and let $L_0$ be 
the set of its leaves. 
Consider the subclass $\cC_0 = \{X_0 \cup L:\ L \seq L_0\}$
of $\cC_{con}(G)$. The instances in $X_0$ or outside $X_0 \cup L_0$
do not distinguish between the concepts in $\cC_0$. But $\cC_0$ restricted 
to the subdomain $L_0$ equals $2^{L_0}$. Hence 
\[ 
\RTD(\cC_{con}(T)) \ge \TD_{min}(\cC_0) = |L_0| = \ell(T)
\enspace .
\]
We next show that $\ell(T)$ bounds $\RTD(\cC_{con}(T)) = \PBTD(\cC_{con}(T))$ 
from above. A nonempty concept in $\cC_{con}(T)$ is a nonempty connected 
subset $X$ of the vertices of $T$ so that $T' := \Span{X}_T$ is a connected 
subgraph of $T$. The third assertion in Remark~\ref{rem:con} implies that $T'$ 
is a tree itself.
Moreover $\ell(T') \le \ell(T)$ and $X$ is the unique smallest concept 
in $\cC_{con}(T)$ that contains all leaves of $T'$. Hence $X$ can be taught 
by having subset-preferences and by presenting all leaves of $T'$ 
as positive examples. Hence $X$ is taught by presenting $\ell(T') \le \ell(T)$ 
examples. Thanks to subset-preferences, the empty set is successfully taught
even if its teaching set is empty.
\item
Fix a maximum-leaf spanning tree $T$ of $G$. The first assertion in
Remark~\ref{rem:con} combined with the (already proven) first assertion
of this theorem implies that
\[
\RTD(\cC_{con}(G)) \ge \RTD(\cC_{con}(T)) = \ell(T) = \ell(G) \enspace .
\]
\item
Let $G_1,\ldots,G_r$ be the components of $G$.
Pick an index $i^* \in [r]$ that is a maximizer of $\RTD(\cC_{con}(G_i))$.
Then
\[
\RTD(\cC_{con}(G)) \stackrel{(\ref{eq:components-rtd})}{\ge} 
\max_{i=1,\ldots,k}\RTD(\cC_{con}(G_i))  \ge 
\max_{i=1,\ldots,k}\ell(G_i) = \ell(G) \enspace .
\] 
We finally show that $\ell(G)+1$ 
bounds $\RTD(\cC_{con}(G)) = \PBTD(\cC_{con}(G))$ from above.
Let's have superset-preferences on the concept class $\cC_{con}(G)$.
Fix a concept $X \in \cC_{con}(G)$. Teach $X$ by presenting 
an (arbitrarily chosen) element of $X$ as a positive example
and all elements in $N_G^o(X)$ as negative examples. Since $N_G^o(X)$
separates the vertices in $X$ from all vertices outside $N_G(x)$,
the resulting version space consists of the set $X$ and all its
subsets. Thanks to superset-preferences, $X$ is the unique most
preferred concept in the version space. The size of our teaching set
for $X$ equals $1 + |N_G^o(X)|$.  It suffices to show 
that $|N_G^o(X)| \le \ell(G)$. There is a unique
index $k \in [r]$ such that $X \in \cC_{con}(G_k)$.
Clearly $N_G^o(X) = N_{G_k}^o(X)$. We know from
Corollary~\ref{cor:tree2spanning-tree} that $|N_{G_k}^o(X)| \le \ell(G_k)$.
Clearly $\ell(G_k) \le \ell(G)$. 
\end{enumerate}
\end{proof}

The following result characterizes the graphs $G$ with the property
that $\VCD(\cC_{con}(G)) = \ell(G)+1$. A similar characterization
is already contained in~\cite{KKRUW1997}.

\begin{lemma} \label{lem:con-vcd}
For any connected graph $G = (V,E)$, the following assertions are equivalent:
\begin{enumerate}
\item 
There exists a maximum-leaf spanning tree $T^*$ of $G$ and 
and an interior vertex $u$ of $T^*$ such that each pair of leaves in $T^*$ 
can be connected by a path in $G$ which does neither contain $u$ nor any
other leaf in $T^*$.
\item
There exists a tree $T \le G$ with $\ell(G)$ leaves and with an interior
vertex $u$ such that each pair of leaves in $T$ can be connected by a path
in $G$ which does neither contain $u$ nor any other leaf in $T$.
\item
$\VCD(\cC_{con}(G)) = \ell(G)+1$.
\end{enumerate}
\end{lemma}

\begin{proof}
The first assertion clearly implies the second-one. An inspection 
of the third assertion in Lemma~\ref{lem:tree2spanning-tree} reveals 
that the second assertion also implies the first-one. We next show 
that the second assertion implies the third-one. This implication
is the special case $k = \ell(G)$ of the following more general result:
\begin{description}
\item[Claim:]
Suppose that there exists a tree $T \le G$ with $k$ leaves and with an interior
vertex $u$ such that each pair of leaves in $T$ can be connected by
a path in $G$ that does neither contain $u$ nor any other leaf in $T$.
Then $\VCD(\cC_{con}(G)) \ge k+1$.
\item[Proof of the Claim:]
Let $L$ be the set of leaves in $T$. It suffices 
to show that the set $L\cup\{u\}$ is shattered by $\cC_{con}(G)$. 
More concretely, if $S \seq L$ is an arbitrary but fixed subset of $L$, 
we must find connected sets $X,X' \seq V$ such that 
\[
X \cap (L\cup\{u\}) = S\ \mbox{ and }\ X' \cap (L\cup\{u\}) = S\cup\{u\}
\enspace .
\]  
The set $X'$ can be simply chosen as the set of all vertices in $T$ 
except for the leaves in $L \sm S$. If $S = \eset$, we may set $X := \eset$.
Suppose now that $S$ is a nonempty subset of $L$.
Choose $X$ now as the vertex set of the subgraph of $G$ consisting 
of all paths $P_{v,w}$ where $v,w$ are (not necessarily distinct) 
vertices in $S$ and $P_{v,w}$ is a path connecting $v$ and $w$ 
but excluding the vertex $u$ as well as any vertex in $L \sm \{v,w\}$.
We conclude from this discussion that the set $L \cup \{u\}$ is indeed 
shattered by $\cC_{con}(G)$ so that $\VCD(\cC_{con}(G)) \ge k)+1$. 
\end{description}
In case of $k=\ell(G)$, we obtain $\VCD(\cC_{con}(G)) = \ell(G)+1$. \\
The proof of the lemma can be accomplished by showing that the third 
assertion implies the second-one.
Set $\ell = \ell(G)$ and consider $\ell+1$ vertices $v_0,v_1,\ldots,v_\ell$
which are shattered by $\cC_{con}(G)$. Since these vertices are shattered,
there must exist paths $P_i$ for $i=1,\ldots,k$ which connect $v_i$ and $v_0$
but exclude any $v_j$ with $j$ different from either $0$ or $i$.
These paths can be used in the obvious way to build a tree $T$
with root $v_0$ and leaves $v_1,\ldots,v_\ell$. Again by the shattering
argument, it follows that each pair of leaves in $T$ can be connected
by a path in $G$ which does neither contain $v_0$ nor any other leaf in $T$.
Hence  the second assertion is valid, which completes the proof of the lemma. 
\end{proof}

\begin{definition}[Opponent of a Connected Set]
Suppose that $G = (V,E)$ is a graph and $X \in \cC_{con}(G) \sm \{\eset\}$. 
Then $Y \in \cC_{con}(G)$ is called a \emph{maximal opponent of $X$} if $Y$ 
spans a connected component in $\Span{V \sm N_G(X)}_G$. 
Any connected subset of a maximal opponent of $X$ is called 
an \emph{opponent of $X$}. 
\end{definition}

\begin{lemma} \label{lem:opponents}
Let $G$ be a graph, let $X \in \cC_{con}(G) \sm \{\eset\}$, and 
let $Y$ be a maximal opponent of $X$. Then $N_G^o(Y) \seq N_G^o(X)$.
Moreover, if $X$ and $Y$ are not part of the same component
of $G$, then $N_G^o(Y) = \eset$. 
\end{lemma}

\begin{proof}
Let $G_i = (V_i,E_i)$ with $i=1,\ldots,r$ be the components of $G$.
There are unique indices $j,k \in [r]$ such that $X \seq V_j$
and $Y \in V_k$. If $k \neq j$ then, since $Y$ is a maximal opponent,
we must have $Y = V_k$ so that $N_G^o(Y) = \eset$.
Suppose now that $k=j$.
Since $N_G^o(X)$ separates $X$ from $Y$, we have $N_G^o(Y) \cap X = \eset$.
Since $Y$ spans a complete connected component in $\Span{V \sm N_G(X)}_G$ 
and $V \seq V_k$, the set $Y$ must span a complete connected component 
in $\Span{V_k \sm N_G(X)}_G$. Thus $N_G^o(Y) \cap (V_k \sm N_G(X)) = \eset$. 
Hence $N_G^o(Y) \seq N_G^o(X)$.
\end{proof}

\begin{theorem} \label{th:con-rtd-vcd}
For any graph $G$: $\RTD(\cC_{con}(G)) \le \VCD(\cC_{con}(G))$.
\end{theorem}

\begin{proof}
If $\VCD(\cC_{con}(G)) = \ell(G)+1$, 
then $\RTD(\cC_{con}(G)) \le \VCD(\cC_{con}(G))$. This is  simply 
because $\ell(G)+1$, as we know already, is a general upper bound
on $\RTD(\cC_{con}(G))$. We may therefore assume 
that $\VCD(\cC_{con}(G)) = \ell(G)$.  Consider the following commitments 
in the preference-based setting:
\begin{enumerate}
\item 
We have superset-preferences on the concept class $\cC_{con}(G)$.
\item
If $X \neq Y \in \cC_{con}(G)$ such that, first, neither $X \subset Y$
nor $Y \subset X$ and, second, $|N_G^o(X)| > |N_G^o(Y)|$, then $X$
is preferred over $Y$. In other words, among incomparable connected
sets, the ones with the largest open neighborhood are preferred most.
\item
A concept $X \in \cC_{con}(G)$ with $|N_G^o(X)| < \ell(G)$
is taught be presenting the vertices in $N_G^o(X)$ as negative
examples and an arbitrarily chosen vertex $u \in X$ as a positive
example.
\item
A concept $X \in \cC_{con}(G)$ with $|N_G^o(X)| = \ell(G)$
is taught be presenting the vertices in $N_G^o(X)$ as negative
examples.
\end{enumerate}
It is obvious that all teaching sets are of size at most $\ell(G)$.
It is furthermore obvious that, thanks to superset-preferences,
a teaching set $N_G^o(X)$ for $X$augmented by a single instance in $X$
makes $X$ the unique most preferred concept in the version space. 
Consider now a concept $X \in \cC_{con}(G)$ with $\ell(G)$
vertices in its open neighborhood. Thanks to superset-preferences,
the most preferred concepts in the version space are $X$ and its
maximal opponents. Recall from Lemma~\ref{lem:opponents} that
the open neighborhood of any maximal opponent $Y$ of $X$ is a subset
of $N_G^o(X)$. 
\begin{description}
\item[Claim:]  
For each maximal opponent $Y$ of $X$, we have $N_G^o(Y) \subset N_G^o(X)$. 
\item[Proof of the Claim:]
Let $G_i = (V_i,E_i)$ with $i=1,\ldots,r$ be the components of $G$.
Assume for sake of contradiction that there exists a maximal opponent $Y$ 
of $X$ such that $N_G^o(Y) = N_G^o(X)$. Then $X$ and $Y$ must belong
to the same component, say $X,Y \seq V_k$.
According to Lemma~\ref{lem:tree2spanning-tree} there exists 
a tree $T_X \le G_k$ that is a spanning tree of $\Span{N_{G_k}(X)}_{G_k}$ 
and has the $\ell(G) = \ell(G_k)$ vertices of $N_G^o(X) = N_{G_k}^o(X)$ among 
its leaves (so that there can be no other leaves in $T_X$). Let $T_Y$ denote
the corresponding tree for $Y$. Since $N_{G_k}^o(X)$ separates $X$
from $Y$, the trees $T_X$ and $T_Y$ have no interior vertex in common.
Choose two arbitrary leaves $z_1$ and $z_2$ of $T_X$. Then the path
from $z_1$ to $z_2$ in $T_Y$ does neither contain an interior vertex
of $T_Y$ nor another leaf of $T_X$. We may conclude from
Lemma~\ref{lem:con-vcd} that $\VCD(\cC_{con}(G_k)) = \ell(G_k)+1 = \ell(G)+1$.
Since $\VCD(\cC_{con}(G)) \ge \VCD(\cC_{con}(G_k))$, this is in contradiction 
with our initial assumption that $\VCD(\cC_{con}(G)) \le \ell(G)$.
\end{description}
Since $N_G^o(X)$ separates $X$ from $Y$, we have $X \cap Y = \eset$.
This implies that none of $X$ and $Y$ is a subset of the other-one.
Thanks to $N_G^o(Y) \subset N_G^o(X)$, the concept $X$ is preferred
over $Y$. Since this reasoning applies to each maximal opponent of $X$,
we may conclude that $X$ is the unique most preferred concept
in the version space induced by the (negative) examples in $N_G^o(X)$.
\end{proof}

\noindent
The preceding results imply that, for any graph $G$, we have that
\[
\ell(G) \le \RTD(\cC_{con}(G)) \le \VCD(\cC_{con}(G)) \le \ell(G)+1
\enspace .
\]
Exactly one of these inequalities must be strict. The examples below
(including the \hbox{graphs} $K_n$, $P_n$ and $C_4$ that we already discussed
in connection with $(\Delta,\RTD,\VCD)$-triples) demonstrate that the 
strict inequality can be in any of the three possible positions.

\begin{examples} \label{ex:con}
The \emph{$(\ell,\RTD,\VCD)$-triple} of a connected graph $G$ 
is defined as the triple 
\[ 
(\ell(G), \RTD(\cC_{con}(G)), \VCD(\cC_{con}(G)))
\enspace .
\]
Note that $\ell(K_n) = n-1 = \Delta(K_n)$ 
and $\cC_{con}(K_n) = 2^{[n]} \sm \{\eset\} = \cC_{star}(K_n)$.
Therefore the $(\ell,\RTD,\VCD)$-triple of $K_n$ equals
the $(\Delta,\RTD,\VCD)$-triple of $K_n$ and, as we know
already, the latter equals $(n-1,n-1,n-1)$. \\
Consider now paths $P_n$ with $n\ge2$. Clearly\footnote{The two 
endpoints are the leaves.} $\ell(P_n)=2$, 
$\cC_{con}(P_n) = \{[i:j]: 1 \le i \le j \le n\}$
and $\VCD(\cC_{con}(P_n)) = 2$. It follows 
that $\RTD(\cC_{con}(P_n)) = 2$ as well.
Hence the $(\ell,\RTD,\VCD)$-triple of $P_n$ equals $(2,2,2)$. \\
We denote the vertices of $C_4$ again by $a,b,c,d$ (in this circular
order). Clearly $\ell(C_4) = 2$, 
and $\cC_{con}(C_4) = 2^{\{a,b,c,d\}} \sm \{ \eset,\{a,c\},\{b,d\} \}$,
which is a concept class of size $16-3 = 13$. The VC-dimension 
of this class is at most $3$. From $13 > 11 = 1 + 4 + \binom{4}{2}$
and~(\ref{eq:sauer-rtd}), we infer that its RT-dimension is at least $3$. 
We may therefore conclude that
the $(\ell,\VCD,\RTD)$-triple of $C_4$ equals $(2,3,3)$. \\
Consider finally the graph $G$ that is shown in Fig~\ref{fig:con}.
It consists of a complete binary tree of height $2$ 
(root $a$, interior vertices $a,b,c$, leaves $d,e,f,g$)  
plus four additional edges (shown in red color) 
which connect $\{d,e\}$ completely bipartitely with $\{f,g\}$.
It is easy to see that any pair of leaves can be connected in $G$
by a path which contains neither the root $a$ nor another leaf. 
For instance, the path $d,b,e$ connects the leaves $d$ and $e$, 
and the leaves $d$ 
and $f$ are directly connected by one of the $4$ additional edges. 
We may therefore apply the claim within the proof of Lemma~\ref{lem:con-vcd} 
and conclude that the set $\{a,d,e,f,g\}$ is shattered 
so that $\VCD(\cC_{con}(G)) \ge 5$. 
We claim that $\RTD(\cC_{con}(G)) \le 4$. This would directly imply
that the $(\ell,\RTD,\VCD)$-triple equals $(4,4,5)$. 
In order to prove the claim, we will show show that each connected set 
in $G$ has a teaching set of size at most $4$ with respect to the 
simple preference relation of preferring sets of larger size over sets 
of smaller size:
\begin{itemize}
\item
For a connected set $X$ of size $3$ or more, we can simply present 
all vertices outside $X$ (at most $4$ many) as negative examples.
\item
For singleton sets $\{x\}$, we can present $x$ as a positive example
and all vertices adjacent to $x$ as negative examples. 
Since $\Delta(G) = 3$, the size of this teaching set is bounded by $4$.
\item
The connected sets of size $2$ (the only remaining case) are
equal to the edges of $G$.
By a straightforward symmetry argument, it suffices to consider
the edges $\{a,b\}$, $\{b,d\}$ and $\{d,f\}$. 
We will choose the respective open neighborhood as the teaching set
(containing negative examples only). The following table 
shows for each of these edges the respective open neighborhood 
(= teaching set) and the respective set of opponents:
\[
\begin{array}{||c|c|c||}
\hline
\mbox{edge}& \mbox{open neighborhood} & opponent(s) \\
\hline
\hline
\{a,b\} & \{c,d,e\} & \{f\} , \{g\} \\
\hline
\{b,d\} & \{a,e,f,g\} & \{c\} \\ 
\hline
\{d,f\} & \{b,c,e,g\} & \{a\} \\
\hline
\end{array}
\]
As can be seen from this table, each of the three considered edges has
only (one or two) opponents of size $1$. Since we prefer sets of larger 
size over sets of smaller size, each of the three edges will be preferred
over its opponent(s) and over its connected proper (singleton) subsets.
\end{itemize}
We conclude from this discussion that $\RTD(\cC_{con}(G)) \le 4$.
\end{examples}

\begin{figure}[hbt]

        \begin{center}
\begin{tikzpicture}


        \tikzset{every node/.style={draw,circle}}

        \def\x{0} \def\y{0}
        \node (a) at (\x,\y) {$a$};
        \node (b) at (\x-2,\y-1) {$b$};
        \node (c) at (\x+2,\y-1) {$c$};
        \node (d) at (\x-3,\y-4) {$d$};
        \node (e) at (\x-1,\y-2) {$e$};
        \node (f) at (\x+1,\y-2) {$f$};
        \node (g) at (\x+3,\y-4) {$g$};

        \draw (a)--(b);
        \draw (a)--(c);
        \draw (b)--(d);
        \draw (b)--(e);
        \draw (c)--(f);
        \draw (c)--(g);
        \draw[red] (d)--(f);
        \draw[red] (d)--(g);
        \draw[red] (e)--(f);
        \draw[red] (e)--(g);

\end{tikzpicture}
        \end{center}
        \caption{A graph $G$ with $\RTD(\cC_{con}(G)) = 4$
                 and $\VCD(\cC_{con}(G)) = 5$.  \label{fig:con}}
\end{figure}
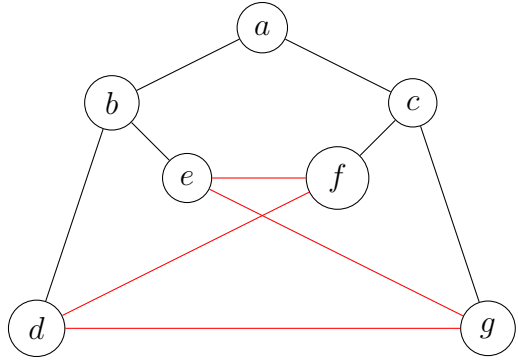

\paragraph{Directions of Future Research.}

Clarify whether the RTD-conjecture can be confirmed for concept classes 
induced by graphs in a way where concepts represent paths, cycles, cliques 
or other subgraphs of the given graph $G$. \\
Explore whether the Sample-Compression conjecture can be confirmed for 
graph-induced concept classes.

\bibliographystyle{plain} 

\begin{thebibliography}{10}

\bibitem{CCT2016}
Xi~Chen, Yu~Cheng, and Bo~Tang.
\newblock On the recursive teaching dimension of {VC} classes.
\newblock In {\em {Proceedings of NeurIPS 2016}}, pages 2164--2171, 2016.

\bibitem{DDSZ2013}
Malte Darnst{\"a}dt, Thorsten Doliwa, Hans~U. Simon, and Sandra Zilles.
\newblock Order compression schemes.
\newblock In {\em {Proceedings of the 24th International Conference on
  Algorithmic Learning Theory (ALT 2013)}}, pages 173--187, 2013.

\bibitem{DFSZ2014}
Thorsten Doliwa, Gaojian Fan, Hans~Ulrich Simon, and Sandra Zilles.
\newblock Recursive teaching dimension, {VC}-dimension and sample compression.
\newblock {\em Journal of Machine Learning Research}, 15:3107--3131, 2014.

\bibitem{GRSZ2017}
Ziyuan Gao, Christoph Ries, Hans~U. Simon, and Sandra Zilles.
\newblock Preference-based teaching.
\newblock {\em Journal of Machine Learning Research}, 18(31):1--32, 2017.

\bibitem{GK1995}
Sally~A. Goldman and Michael~J. Kearns.
\newblock On the complexity of teaching.
\newblock {\em Journal of Computer and System Sciences}, 50(1):20--31, 1995.

\bibitem{HWLW2017}
Lunija Hu, Ruihan Wu, Tianhong Li, and Liwei Wang.
\newblock Quadratic upper bound for recursive teaching dimension of finite {VC}
  classes.
\newblock In {\em Proceedings of Machine Learning Research}, volume~65, pages
  1147--1156, 2017.

\bibitem{KKRUW1997}
Evangelos Kranakis, Danny Krizanc, Berthold Ruf, Jorge Urrutia, and Gerhard
  Woeginger.
\newblock The {VC}-dimension of set systems defined by graphs.
\newblock {\em Discrete Applied Mathematics}, 77(3):237--257, 1997.

\bibitem{SSYZ2014}
Rahim Samei, Pavel Semukhin, Boting Yang, and Sandra Zilles.
\newblock Algebraic methods proving {S}auer's bound for teaching.
\newblock {\em Theoretical Computer Science}, 558:35--50, 2014.

\bibitem{S1972}
N.~Sauer.
\newblock On the density of families of sets.
\newblock {\em Journal of Combinatorial Theory, Series A}, 13(1):145--147,
  1972.

\bibitem{SH1972}
S.~Shelah.
\newblock A combinatorial problem: Stability and order for models and theories
  in infinitary languages.
\newblock {\em Pacific Journal of Mathematics}, 41:247--261, 1972.

\bibitem{SM1991}
Ayumi Shinohara and Satoru Miyano.
\newblock Teachability in computational learning.
\newblock {\em New Generation Computing}, 8(4):337--348, 1991.

\bibitem{SZ2015}
Hans~U. Simon and Sandra Zilles.
\newblock Open problem: Recursive teaching dimension versus {VC} dimension.
\newblock In {\em Proceedings of the 28th Annual Conference on Learning
  Theory}, pages 1770--1772, 2015.

\bibitem{VC1971}
Vladimir~N. Vapnik and Alexey~Ya. Chervonenkis.
\newblock On the uniform convergence of relative frequencies of events to their
  probabilities.
\newblock {\em Theory of Probability \& its Applications}, 16(2):264--280,
  1971.

\bibitem{W2003}
Manfred Warmuth.
\newblock Open problem: Compressing to {VC}-dimension many points.
\newblock In {\em Proceedings of the 18th Annual Conference on Computational
  Complexity}, pages 743--744, 2003.

\bibitem{ZLHZ2011}
Sandra Zilles, Steffen Lange, Robert Holte, and Martin Zinkevich.
\newblock Models of cooperative teaching and learning.
\newblock {\em Journal of Machine Learning Research}, 12:349--384, 2011.

\end{thebibliography}

\end{document}